\documentclass[10pt,reqno]{article}
\usepackage{amsmath,amssymb,amsthm}
\usepackage{esint}
\usepackage{graphicx,color}
\usepackage{algorithm}
\usepackage{algpseudocode}
\usepackage[sort&compress]{natbib}
\newcommand{\nm}{\medskip}
\DeclareMathOperator*{\argmin}{arg\,min}
\DeclareMathAlphabet{\itbf}{OML}{cmm}{b}{it}

\def\by{{{\itbf y}}}
\def\bx{{{\itbf x}}}

\def\bz{{{\itbf z}}}

\newcommand{\G}{\mathbf{G}}
\newcommand{\eps}{\epsilon}
\newcommand{\I}{{\mathbf{I}}}
\newcommand{\Ic}{\mathcal{I}}

\newcommand{\E}{\mathbf{E}}
\newcommand{\h}{\mathbf{H}}
\newcommand{\J}{{\mathbf{J}}}

\newcommand{\RR}{\mathbb{R}}

\newcommand{\bd}{\mathbf{d}}
\newcommand{\bdw}{\mathbf{d}_\W}

\newcommand{\bp}{\mathbf{p}}
\newcommand{\q}{\mathbf{q}}

\newcommand{\K}{{\kappa}}

\newcommand{\OL}{\mathcal{L}}
\newcommand{\bu}{\mathbf{u}}
\newcommand{\bv}{\mathbf{v}}

\newcommand{\W}{\mathcal{W}}

\newcommand{\ds}{\displaystyle}
\newcommand{\p}{\partial}

\newcommand{\hG}{\widehat{\G}}

\newcommand{\OM}{\mathcal{M}}
\newcommand{\OP}{{\mathcal{P}}}
\newcommand{\OT}{\mathcal{T}}
\newcommand{\OR}{\mathcal{R}}
\newcommand{\B}{\mathbf{B}}
\newcommand{\C}{\mathcal{C}}

\newtheorem{thm}{Theorem}[section]

\newtheorem{lem}[thm]{Lemma}

\numberwithin{equation}{section}

\newcommand{\pathfigures}{Figures/}
\graphicspath{{\pathfigures}}

\begin{document}
\title{
Electromagnetic source localization with finite set of frequency measurements
}
\author{
Abdul Wahab \thanks{Address correspondence to A. Wahab, E-Mail: wahab@ciitwah.edu.pk, Phone: +92-(51)-9272614.} \thanks{\footnotesize Department of Mathematics, COMSATS Institute of Information Technology, 47040, Wah Cantt., Pakistan. (wahab@ciitwah.edu.pk, amerasheed@ciitwah.edu.pk, shum.88@yahoo.com)}
\and  Amer Rasheed\footnotemark[2]
\and  Rab Nawaz\thanks{Department of Mathematics, COMSATS Institute of Information Technology, Park Road, Chak Shahzad, 44000, Islamabad, Pakistan. (rabnawaz@comsats.edu.pk)} 
\and  Saman Anjum\footnotemark[2] 
}
\maketitle

\begin{abstract}
A phase conjugation algorithm for localizing an extended radiating electromagnetic source from boundary measurements of   the electric field is presented. Measurements are taken over a finite number of frequencies. The artifacts related to the finite frequency data are tackled with $l_1-$regularization blended with the \emph{fast iterative shrinkage-thresholding algorithm with backtracking} of Beck \& Teboulle.
\end{abstract}



\section{Introduction}

Inverse source problems have been the subject of numerous studies over the recent past due to a plethora of applications in science and engineering, specially in biomedical imaging,  non destructive testing and prospecting geophysics; see, for instance, \cite{HAetal-11, HAetal-11b, noise, PAT, aggk, book1, book2, Vald, EEG-rev, Porter, jvc} and references therein. Several frameworks to recover spatial and temporal support of the acoustic, elastic and electromagnetic sources in time and frequency domain have been developed \cite{noise, PAT, Vald, source, EMTr}, including time reversal and phase conjugation  algorithms \cite{HAetal-11, HAetal-11b, aggk, Carminati, fink, GG, GWL}.

Time reversal algorithms and their frequency domain counterparts the phase conjugation algorithms have observed a significant success in the resolution of inverse problems in assorted disciplines of science since their first applications \cite{fink, book, josselin}. These algorithms exploit the self-adjointness and the reciprocity of non-attenuating waves. This substantiates that a wave  retraces its path backwards in chronology through a loss-less medium and converges at the location of its source (scatterer, reflector or emitter) on re-emission after reversing the time using transformation $t\to t_{final}-t$ or using phase conjugation in frequency domain \cite{fink, Carminati}. 

Their simplicity and robustness motivates the application of time reversal and phase conjugation algorithms to deal with source localization problems. If the sources are spatially  punctual (Dirac delta sources), a single measurement of the emitted wave over an imaging surface on single time or frequency is sufficient to locate the source position. However, when the sources are extended in space, the problem is well-posed when the available data at hand is computed over an interval of time $[0,T]$ for sufficiently large $T>0$ or over a large frequency bandwidth. The data collected at each time or frequency, contain useful information of an extended source that can be retrieved on re-emission of the data after time reversal or phase conjugation. The superposition of all the retrieved information provides the spatial or temporal support of the extended source with a resolution bounded by Raleigh diffraction limit \cite{HAetal-11b, HAetal-11, aggk}.  

In real physical configurations a complete measurement set over an entire frequency or time interval is not feasible. The only measurements available are made over a finite set of frequencies or time instances. The lack of information at hand, therefore induces noise in the reconstruction of the extended source and makes it challenging to reconstruct the support of the source stably and accurately with good resolution properties. All one can get is an initial guess of the support of the source.

This work aims to evince a radiating source for the Maxwell's equations using boundary measurements over a finite set of frequencies using a phase conjugation sensitivity framework blended with \emph{fast iterative shrinkage thresholding algorithm with backtracking} of \citet{Beck_teboulle} for $l_1-$regularization; refer also to \citet{PAT}. The phase conjugation algorithm yields an initial guess of the spatial support of an extended source whereas the $l_1-$ regularization optimizes the initial guess. Since the $l_1-$regularization term is not smooth, this motivates the application of iterative shrinkage thresholding algorithm of \citet{Beck_teboulle}.  This method belongs to the class of split gradient descent iterative schemes with a global convergence rate $O(k^{-2})$, where $k$ is the iteration counter.

The inverse source problems are ill-posed having non-uniqueness issues generally due to the presence of non radiating sources \cite{Bleistein, Porter, Vald}. The stability and localization of electromagnetic radiating sources with single, multiple and entire frequency or time data have been extensively studied in, for instance,  \cite{Porter, Bleistein, Albanese, bao, Bojarski}. The well-posedness of problem undertaken in this study, that is, the source problem for Maxwell's equation with a finite set of frequency surface measurements is established in particular in \cite{Vald}.

The investigation is sorted in the following order. The inverse problem is presented and a few key identities are  collected in Section \ref{s:problem}. In Section \ref{s:full}, an electromagnetic source is retrieved using entire frequency measurements with a phase conjugation functional. The functional is further adopted to locate a source using finite set of frequencies in Section \ref{s:fin}. An initial guess is retrieved and then optimized using $l_1-$regularization. The principle contributions of the investigation are summarized in Section \ref{s:conc}.

\section{Problem Formulation}\label{s:problem}
Let $\Omega\subset\RR^3$ be an open bounded domain with a Lipschitz boundary $\Gamma$. Consider %the time-harmonic Maxwell's equations
\begin{equation}
\begin{cases}
\nabla\times\E-i\omega\mu_0\h =\mathbf{0}, &\bx\in\RR^3,
\\
\nabla\times\h+i\omega\eps_0\E= \J(\bx), & \bx\in\RR^3,
\end{cases}
\label{EMw}
\end{equation}
subject to the Silver-M\"uller radiation conditions
\begin{equation}
\ds\lim_{|\bx|\to\infty} |\bx| 
%\begin{cases}
\Big(\sqrt{\mu_0}\h\times\hat{\bx}-\sqrt{\eps_0}\E\Big) =\mathbf{0},
%\\
%\sqrt{\eps_0}\E(\bx,\omega)\times\hat{\bx}+\sqrt{\mu_0}\h,
%\end{cases}
\qquad\text{where}\quad\hat{\bx}:=\frac{\bx}{|\bx|},
\label{radiation}
\end{equation}
with frequency pulsation $\omega$, electric permittivity $\eps_0>0$ and magnetic permeability $\mu_0>0$, where $\E$ and $\h$ are the time-harmonic electric and magnetic fields respectively. Here $\J(\bx)\in\RR^3$ is the current  source density, assumed to be sufficiently smooth and compactly supported in  $\Omega$, that is, $\rm{supp}\big\{\J\big\}\subset\subset\Omega$.

Define the admissible set of frequencies and the boundary data consisting of the electric field respectively by
\begin{equation}
%\W:=\big\{\omega_n\in\RR:\quad 1\leq n\leq N,\quad n,N\in\NN\big\}
\W:=\big(\omega_n\big)_{n=1}^N
\quad\text{and}\quad \bd(\bx,\omega)= \E(\bx,\omega), \quad\forall (\bx,\omega)\in\Gamma\times\RR.
\label{Wd}
\end{equation}
Then, the ultimate goal of this work is to tackle the following problem: 

\begin{proof}[Statement of Problem] 
Given $\bdw:=\bd|_{\Gamma\times\W}$ for $N$ sufficiently large, identify the support, $\rm{supp}\big\{\J\big\}$, of current source density $\J$.
\end{proof}

\subsection{Preliminaries}\label{ss:pre}

In the sequel, we refer to $\K_0:=\omega\sqrt{\eps_0\mu_0}= {\omega}/{c_0}$ as the wave number with $c_0:= {1}/{\sqrt{\eps_0\mu_0}}$ being the wave speed in dielectrics and let $\nu$ be the outward unit normal to $\Gamma$. 
By virtue of \eqref{EMw}, the time-harmonic electric field satisfies the Helmholtz equation
\begin{equation}
\nabla\times\nabla\times\E- \K_0^2\E= i\omega\mu_0\J(\bx),\quad \bx\in\RR^3,
\label{Helm}
\end{equation}
subject to the outgoing radiation condition  \eqref{radiation}.

Let $\G_0^{\rm{ee}}(\bx,\omega)$ be the outgoing electric-electric Green's function for the Maxwell's equations \eqref{EMw} in $\RR^3$, that is, 
\begin{equation}
\nabla\times\nabla\times\G_0^{\rm{ee}}(\bx,\omega)-\K_0^2\G_0^{\rm{ee}}(\bx,\omega)
=i\omega\mu_0\I\delta_{\mathbf 0}(\bx), \quad \bx\in\RR^3,
\label{GeW}
\end{equation}
where  $\delta_{\bf{0}}(\bx)$ is the Dirac mass at $\bx=\mathbf{0}$. 
%It is well-known, see for instance \cite[Theorem 5.2.1, 5.2.2]{nedelec} and \cite{Hansen}, that there exists a unique solution $\G_0^{\rm{ee}}$ to \eqref{GeW} given by 
%\begin{equation}
%\ds\G_0^{\rm{ee}}(\bx,\omega)=-i\omega\mu_0\left(\I+\dfrac{1}{\K_0^2}\nabla\nabla\right)g(\bx,\omega),
%\end{equation}
%where we refer to $g(\bx,\omega)$ as the fundamental solution to the Helmholtz operator $-(\Delta+\K_0^2)$ in $\RR^d$, subject to the Sommerfeld's outgoing radiation conditions. Recall that $g(\bx,\omega)$ is given by 
%\begin{equation}
%g(\bx,\omega)=
%\begin{cases}
%\ds-\dfrac{1}{4i}H_0^{(1)}(\K_0|\bx|), & \qquad \bx\neq 0, \bx\in\RR^2,
%\\
%\ds\dfrac{1}{4\pi|\bx|}\exp\{i\K_0|\bx|\}, &\qquad \bx\neq 0,  \bx\in\RR^3,
%\end{cases}
%\end{equation}
%where $H^{(1)}_0$ is the zeroth order Hankel function of first kind \cite{book, nedelec}.
%To end this section, 
Following spatial reciprocity can be proved for isotropic dielectrics; see \cite{Wapenaar}:
\begin{eqnarray}
\G_0^{\rm{ee}}(\bx-\by,\omega)=\G_0^{\rm{ee}}(\by-\bx,\omega),\qquad\bx,\by\in\RR^3,\quad  \omega\in\RR.
\label{reciprocity}
\end{eqnarray}

\subsection{Electromagnetic Identities}\label{ss:identities}
The following identities are the key ingredients to elucidate the localization property of the imaging functional proposed in the next section. The variants of the identity in Lemma \ref{EM-HKI} can be found in literature; see, for instance, \cite{Carminati, chen}. The details are, however, provided here  since it is one of the building blocks of our reconstruction algorithm.

\begin{lem}[EM Helmholtz-Kirchhoff Identity]\label{EM-HKI}
Let $\B(0,R)$ be an open ball in $\RR^3$ with large radius $R\to\infty$ and boundary $\partial \B(0,R)$. Then, for all $\bx,\by\in\RR^3$, we have 
$$
\ds\lim_{R\to +\infty}\int_{\partial \B(0,R)}\G_0^{\rm{ee}}(\bx-\xi,\omega)\overline{\G_0^{\rm{ee}}}(\xi-\by,\omega)d\sigma(\xi)=\mu_0 c_0\Re e\bigg\{\G_0^{\rm{ee}}(\bx-\by,\omega)\bigg\},
$$
where the superposed bar indicates a complex conjugate.
\end{lem}

\begin{proof}
Recall that we have for all constant vectors $\bp, \q\in\RR^3$, and $\bx,\by\in\RR^3$,
\begin{eqnarray}
\nabla_\xi \times \nabla_\xi \times \G_0^{\rm{ee}}(\bx-\xi,\omega)\bp -\K_0^2\G_0^{\rm{ee}}(\bx-\xi,\omega)\bp&=&i\omega\mu_0\delta_\bx(\xi)\bp, \label{alpha}
\\
\nabla_\xi \times \nabla_\xi \times \overline{\G_0^{\rm{ee}}}(\by-\xi,\omega)\q-\K_0^2 \overline{\G_0^{\rm{ee}}}(\by-\xi,\omega)\q&=&-i\omega\mu_0\delta_\by(\xi)\q. \label{beta} 
\end{eqnarray}
 
Taking scalar product of \eqref{alpha} by $\overline{\G_0^{\rm{ee}}}(\by-\xi,\omega)\q$ and  \eqref{beta} by $\G_0^{\rm{ee}}(\bx-\xi,\omega)\bp$, %we get 
%\begin{equation} 
%\nabla_\xi \times \nabla_\xi\times\G_0^{\rm{ee}}(\bx-\xi,%\omega)\bp\cdot \overline{\G_0^{\rm{ee}}}(\by-\xi,\omega)\q
%%\nonumber
%%\\
%%\nm
%%&&\qquad\quad 
%-\K_0^2\G_0^{\rm{ee}}(\bx-\xi,\omega)\bp\cdot %\overline{\G_0^{\rm{ee}}}(\by-\xi,\omega)\q
%= i\omega\mu_0\delta_\bx(\xi)\bp\cdot\overline{\G_0^{\rm{ee}}}
%(\by-\xi,\omega)\q, \label{alpha_2}
%\end{equation}
%\begin{equation} 
%\nabla_\xi \times \nabla_\xi \times \overline{\G_0^{\rm{ee}}}
%(\by-\xi,\omega)\q\cdot\G_0^{\rm{ee}}(\bx-\xi,\omega)\bp
%%\nonumber
%%\\
%%\nm
%%&&\qquad\quad 
%-\K_0^2\overline{\G_0^{\rm{ee}}}(\by-\xi,%\omega)\q\cdot\G_0^{\rm{ee}}(\bx-\xi,\omega)\bp
%= -i\omega\mu_0\delta_\by(\xi)\q\cdot{\G_0^{\rm{ee}}}(\bx-\xi,%\omega)\bp. \label{beta_2}
%\end{equation}
%Now,
subtracting the resultant equations %\eqref{alpha_2} from \eqref{beta_2}, 
%we obtain
%\begin{eqnarray}
%&&\nabla_\xi \times \nabla_\xi \times \overline{\G_0^{\rm{ee}}}(\by-\xi,\omega)\q\cdot\G_0^{\rm{ee}}(\bx-\xi,\omega)\bp
%\nonumber
%\\
%\nm
%&&\qquad\qquad\qquad\qquad\qquad
%-
%\nabla_\xi \times \nabla_\xi \times\G_0^{\rm{ee}}(\bx-\xi,\omega)\bp\cdot \overline{\G_0^{\rm{ee}}}(\by-\xi,\omega)\q  
%\nonumber
%\\\nm
%&&=
%-i\omega\mu_0\delta_\bx(\xi)\bp\cdot\overline{\G_0^{\rm{ee}}}(\by-\xi,\omega)\q-i\omega\mu_0\delta_\by(\xi)\q\cdot{\G_0^{\rm{ee}}}(\bx-\xi,\omega)\bp
%\end{eqnarray}
%On 
and finally integrating over $\B(0,R)$, we arrive at 
\begin{eqnarray}
&&\int_\B\left(\nabla_\xi \times \nabla_\xi \times \overline{\G_0^{\rm{ee}}}(\by-\xi,\omega)\q\right) \cdot \G_0^{\rm{ee}}(\bx-\xi,\omega)\bp\, d\xi 
\nonumber
\\\nm
&&\qquad\qquad\qquad
-\int_\B\overline{\G_0^{\rm{ee}}}(\by-\xi,\omega)\q \cdot \left(\nabla_\xi \times \nabla_\xi \times \G_0^{\rm{ee}}(\bx-\xi,\omega)\bp\right)d\xi 
\nonumber
\\\nm
&&\quad=
-i\omega\mu_0\int_\B\overline{\G_0^{\rm{ee}}}(\by-\xi,\omega)\q\cdot\delta_\bx(\xi)\bp d\xi-i\omega\mu_0\int_{\B}\G_0^{\rm{ee}}(\bx-\xi,\omega)\bp\cdot\delta_\by(\xi)\q d\xi,
\nonumber
\\\nm
&&\quad=
-i\omega\mu_0\bp\cdot\overline{\G_0^{\rm{ee}}}(\by-\bx,\omega)\q-i\omega\mu_0\q\cdot\G_0^{\rm{ee}}(\bx-\by,\omega)\bp
\nonumber
\\\nm 
&&\quad
=
-2i\omega\mu_0\q \cdot\Re e\bigg\{\G_0^{\rm{ee}}(\bx-\by,\omega)\bigg\}\bp.\label{em:rhs}
\end{eqnarray}

On the other hand, recall that, for all $\bu,\bv\in\C^2(\overline{\B})^3$
\begin{eqnarray*}
&&\ds\int_\B{\left(\nabla \times \nabla \times \bu\right) \cdot \bv}dx-\int_\B{\bu \cdot \left(\nabla \times \nabla \times \bv\right)}dx
\\\nm
&&
\qquad\qquad\qquad
\ds=-\int_{\p \B}{(\bu \times\nu) \cdot (\nabla \times \bv)}d\sigma(x)-\int_{\p \B}{(\nabla\times \bu \times \nu)\cdot \bv}d\sigma(x).
\end{eqnarray*}
where $d\sigma$ is the surface element.
Substitute $\bu = \overline{\G_0^{\rm{ee}}}(\by-\xi,\omega)\q$ and  $\bv = \G_0^{\rm{ee}}(\bx-\xi,\omega)\bp$ 
% and take $A=\B(0,R)$ in the above equation, 
to get
\begin{eqnarray*}
&&\int_\B\left(\nabla_\xi \times \nabla_\xi \times \overline{\G_0^{\rm{ee}}}(\by-\xi,\omega)\q\right) \cdot \G_0^{\rm{ee}}(\bx-\xi,\omega)\bp\, d\xi 
\\\nm
&&\qquad\qquad
- \int_\B\overline{\G_0^{\rm{ee}}}(\by-\xi,\omega)\q \cdot \left(\nabla_\xi \times \nabla_\xi \times \G_0^{\rm{ee}}(\bx-\xi,\omega)\bp\right)d\xi 
\\\nm
&&= -\int_{\p \B}\left(\overline{\G_0^{\rm{ee}}}(\by-\xi,\omega)\q \times \nu\right) \cdot \left(\nabla_\xi \times \G_0^{\rm{ee}}(\bx-\xi,\omega)\bp\right)d\xi 
\\\nm
&&
\qquad\qquad
-\int_{\p \B}\left(\nabla_\xi\times \overline{\G_0^{\rm{ee}}}(\by-\xi,\omega)\q \times \nu\right) \cdot \left(\G_0^{\rm{ee}}(\bx-\xi,\omega)\bp\right)d\sigma(\xi),
\\\nm
&&
= -\int_{\p \B}\overline{\G_0^{\rm{ee}}}(\by-\xi,\omega)\q\cdot\left(\nu \times \nabla_\xi \times \G_0^{\rm{ee}}(\bx-\xi,\omega)\bp\right)d\xi
\\\nm
&&
\qquad\qquad
+\int_{\p \B}\left(\nu\times\nabla_\xi\times \overline{\G_0^{\rm{ee}}}(\by-\xi,\omega)\q \right) \cdot \left(\G_0^{\rm{ee}}(\bx-\xi,\omega)\bp\right)d\sigma(\xi).
\end{eqnarray*}

Now from Sommerfeld radiation conditions, 
$$
\nu \times \nabla_\xi \times \G_0^{\rm{ee}}(\bx-\xi,\omega)\bp = i\K_0 \G_0^{\rm{ee}}(\bx-\xi,\omega)\bp+ O\left(R^{-2}\right),
$$
%$$
%\nu \times \nabla_\xi \times \overline{\G_0^{\rm{ee}}}(\by-\xi,\omega)\q = -i\K_0 \overline{\G_0^{\rm{ee}}}(\by-\xi,\omega)\q+ O\left(R^{-\frac{d+1}{2}}\right),
%$$
where the order term vanishes as $R\to\infty$.
Therefore,
\begin{align}
\int_\B &\left(\nabla_\xi \times  \nabla_\xi \times \overline{\G_0^{\rm{ee}}}(\by-\xi,\omega)\q\right)\cdot \G_0^{\rm{ee}}(\bx-\xi,\omega)\bp\, d\xi 
\nonumber
\\\nm
&\quad
- \int_\B\overline{\G_0^{\rm{ee}}}(\by-\xi,\omega)\q \cdot \left(\nabla_\xi \times \nabla_\xi \times \G_0^{\rm{ee}}(\bx-\xi,\omega)\bp\right)d\xi 
\nonumber
\\\nm
&\qquad
\simeq
-2\ds\int_{\p \B}i\K_0\G_0^{\rm{ee}}(\bx-\xi,\omega)\bp\cdot\overline{\G_0^{\rm{ee}}}(\by-\xi,\omega)\q\, d\sigma(\xi)
\nonumber
\\\nm
&\qquad
=
-2i\K_0\int_{\p \B}\q\cdot\overline{\G_0^{\rm{ee}}}(\by-\xi,\omega)\G_0^{\rm{ee}}(\bx-\xi,\omega)\bp d\sigma(\xi),
\label{em:lhs}
\end{align}
where we have used the reciprocity relation \eqref{reciprocity}. 

Finally, comparing Equations \eqref{em:lhs} and \eqref{em:rhs}, %we find out that 
%$$
%\K_0\int_{\p \B}\q\cdot\overline{\G_0^{\rm{ee}}}(\by-\xi,\omega)\G_0^{\rm{ee}}(\bx-\xi,\omega)\bp\, d\sigma(\xi)=\q \cdot\omega\mu_0\Re e\bigg\{\G_0^{\rm{ee}}(\bx-\by,\omega)\bigg\}\bp,
%$$
%for all constant vectors $\bp,\q\in\RR^d$. We deduce 
and by varying and  choosing  $\bp$ and $\q$  as the basis vectors in $\RR^d$ we get
$$
\int_{\p \B}\overline{\G_0^{\rm{ee}}}(\by-\xi,\omega)\G_0^{\rm{ee}}(\bx-\xi,\omega)\, d\sigma(\xi)
\simeq
\mu_0 c_0\Re e\bigg\{\G_0^{\rm{ee}}(\bx-\by,\omega)\bigg\},
$$
which leads to the conclusion by tending $R\to\infty$.
\end{proof}

\begin{lem}\label{lem2}
For all $\bx,\by\in\RR^3$, $\bx\neq\by$,
$
\quad\ds\frac{\epsilon_0}{2\pi}\int_\RR\Re e\bigg\{\G_0^{\rm{ee}}(\bx-\by,\omega)\bigg\}\,d\omega= \delta_\bx(\by)\I.
$
\end{lem}

\begin{proof}
Let  $\widehat{\G}$  be the solution to  
\begin{eqnarray}
\label{g}
\dfrac{1}{c^2_0}\dfrac{\partial^2\widehat{\G}}{\partial t^2}(\bx,t;\by,\tau)+\nabla\times\nabla\times\widehat{\G}(\bx,t;\by,\tau)=-\delta_\by(\bx)\delta_\tau(t)\I, &\bx,\by \in\RR^3, t>\tau,
\end{eqnarray}
and  let $\G(\bx,\by,\omega)$ be the Fourier transform of $\widehat{\G}(\bx,t;\by,0)$. Further, the following causality conditions are imposed
$$
\widehat{\G}(\bx,t;\by,\tau)=0=\dfrac{\partial\widehat{\G}}{\partial t}(\bx,t;\by,\tau),  \qquad \bx,\by\in\RR^3, t<\tau.
$$
Then, integrating \eqref{g} over an infinitesimal time interval from $\tau^-$ to $\tau^+$, using the causality conditions above and the continuity of $\widehat{\G}$ away from $t=\tau$,   we find out that 
\begin{eqnarray}
\dfrac{\partial\widehat{\G}}{\partial t}(x,t;y,\tau)\Big|_{t=\tau^+}=-c^2_0\delta_\by(\bx)\I.
\end{eqnarray}
Therefore, integrating above equation over $t$ and using the Parseval's identity, yields
\begin{eqnarray}
\ds\int_\RR i\omega\G(\bx,\by,\omega)d\omega=2c_0^2\pi\delta_\by(\bx)\I
\int_\RR\delta_0(\omega)d\omega=2c_0^2\pi\delta_\by(\bx)\I,
\label{relation}
\end{eqnarray}
where we have made use of the fact that the Fourier transform of $1$ is $2\pi\delta_0(\omega)$.
Finally, since 
$$
\G_0^{\rm{ee}}(\bx-\by,\omega)=-i\omega\mu_0\G(\bx,\by,\omega),
$$
and $\delta_\by(\bx)\I$ is real, relation \eqref{relation} leads to the conclusion. 
\end{proof}

\section{Reconstruction with Full Bandwidth Measurements}\label{s:full}

This section is in order to provide the building blocks to handle finite frequency data problem. As a first step toward the ultimate goal, we find the spatial support of the current source, ${\rm{supp}}\{\J\}$, from data $\bd(\bx,\omega)$ with $\omega\in\RR$. 

For a fixed frequency $\omega\in\RR$, define the adjoint field $\E^*$ to be the solution to 
\begin{equation}
%\begin{cases}
\nabla\times\nabla\times\E^*(\bx,\omega) -\K_0^2 \E^*(\bx,\omega) = i\omega\mu_0\overline{\bd}(\bx,\omega)\chi_{\Gamma}(\bx), \quad (\bx,\omega)\in \RR^3\times\RR,
%\end{cases}
\label{adjoint}
\end{equation}
where $\chi_\Gamma$ is the characteristic  function of the boundary $\Gamma$. Then, the \emph{phase conjugation} functional is defined by
\begin{equation}
\Ic(\bx):=\ds\frac{\epsilon_0}{2\pi c_0\mu_0}\int_\RR\E^*(\bx,\omega) d\omega,\quad\forall \bx\in\RR^3.
\label{I}
\end{equation}
Then, $\Ic(\bx)$ yields the approximate  spatial support of  the current density $\J(\bx)$. In fact, we have the following theorem.

\begin{thm}\label{thm:full}
For $\bx\in\Omega$ sufficiently far from $\Gamma$ compared to the wavelength of the wave impinging upon $\Omega$, 
$$
\Ic(\bx) \simeq \J(\bx).
$$
\end{thm}
\begin{proof}

Since $\J$ is supported compactly inside $\Omega$, for all $\bx\in\Omega$ and $\by\in\Gamma$, we have
%\begin{equation}
%\E^*(\bx,\omega)=\ds\frac{1}{\mu_0}\int_{\Gamma}\G^{\rm{ee}}_0(\bx-\by,\omega)\overline{\bd}(\by,\omega)d\sigma(\by)
%\end{equation}
\begin{eqnarray}
\begin{cases}
\E^\ast(\bx,\omega)
%&=&\ds\frac{1}{\mu_0}\overline{\bd}(\bx,\omega)\chi_{\Gamma}(\bx)\ast_{\bx} \G^{\rm{ee}}_0(\bx,\omega),\\\nm
%&=&
%\ds\int_{\RR^d}\G^{\rm{ee}}_0(\by-\bx,\omega)\overline{\bd}(\by,\omega)\chi_\Gamma(\by)d\by,
=\ds \int_{\Gamma}\G^{\rm{ee}}_0(\by-\bx,\omega)\overline{\bd}(\by,\omega)d\sigma(\by)
\\
{\bd}(\by,\omega)= \E(\by,\omega)\big|_{\by\in\Gamma}
=\ds\int_{\Omega}\G^{\rm{ee}}_0(\by-\bz,\omega)\J(\bz)d\bz\bigg|_{\by\in\Gamma}.
\end{cases}
\label{EDG-Link}
\end{eqnarray}

Therefore, by using \eqref{EDG-Link} in \eqref{I} we arrive at 
\begin{eqnarray*}
\Ic(\bx) =
%\frac{1}{2\pi  c_0\mu_0^2}\int_\RR\int_{\Gamma}\int_\Omega\G_0^{\rm{ee}}(\bx-\by,\omega)\overline{\G_0^{\rm{ee}}}(\by-\bz,\omega)\J(\bz) d\bz d\sigma(\by)d\omega,
%\\
%&=&
\frac{\epsilon_0}{2\pi  c_0\mu_0}\int_{\RR^3}\int_\RR\left(\int_{\Gamma}\G_0^{\rm{ee}}(\bx-\by,\omega)\overline{\G^{\rm{ee}}_0}(\by-\bz,\omega) d\sigma(\by)\right)d\omega \J(\bz) d\bz.
\end{eqnarray*}

Now, we invoke the Helmholtz-Kirchhoff identity from Lemma \ref{EM-HKI} to have
\begin{equation}
\ds\int_{\Gamma}\G_0^{\rm{ee}}(\bx-\by,\omega)\overline{\G_0^{\rm{ee}}}(\by-\bz,\omega)d\sigma(\by)\simeq\mu_0 c_0\Re e\bigg\{\G_0^{\rm{ee}}(\bx-\bz,\omega)\bigg\},
\end{equation}
leading us to 
\begin{equation*}
\ds\Ic(\bx)\simeq\int_{\RR^3}\left(\frac{\epsilon_0}{2\pi}\int_\RR \Re e\bigg\{\G_0^{\rm{ee}}(\bx-\bz,\omega)\bigg\}d\omega \right)\J(\bz)d\bz.
\end{equation*}

Finally, using Lemma \ref{lem2}, we conclude that %recall that 
%\begin{equation*}
%\ds\frac{1}{2\pi\mu_0}\int_\RR\Re e\bigg\{\G_0^{\rm{ee}}(\bx-\bz,\omega)\bigg\}\,d\omega= \delta_\bx(\bz)\I.
%\end{equation*}
%Therefore, 
\begin{equation*}
\Ic(\bx)\simeq\ds\int_{\RR^3}\delta_\bx(\bz)\J(z)d\bz = \J(\bx).
\end{equation*}
%which concludes the prove.
\end{proof}

\section{Source Localization with Finite Set of Frequencies}\label{s:fin}

In this section, we address the electromagnetic inverse source problem using the boundary data $\bd_\W=\bd(\bx,\omega)\big|_{\Gamma\times\W}$.
%of  the tangential component of  electric field over $\Gamma$ at a finite number of frequencies. 
First an initial guess is retrieved and then subsequently optimized using an $l_1-$regularization technique.

\subsection{Initial Guess Retrieval}\label{ss:initial}

%This section aims to provide a first approximation of the support of current source density $\J$ using a phase conjugation algorithm using the finite set of frequency measurements $\bdw$. 
Inspired by the functional $\Ic$ defined in \eqref{I}, we define a single frequency functional $\Ic_n$ by \eqref{I_n}. However, since we are dealing with a finite set of frequency measurements, the  lack of information over entire spectrum induces noise and blurring in the reconstruction. In order to fix the problem, in this section, an initial guess to the current source density is identified, which will be optimized providing an improved approximation to ${\rm{supp}}\{\J\}$.%, and is regularized in the next section using a total variation Tikhonov regularization blended with an iterative shrinkage tresholding algorithm, thereby providing a good approximation to $\J(\bx)$.

Let $0\leq \K_0^1\leq \K_0^2\leq  \cdots \leq \K_0^N$ be $N$ wave numbers  
%and $\E_n^\ast$ be the phase conjugated waves 
corresponding to $\omega_n\in\W$ for $n=1,2,\cdots, N$. Let us define the adjoint field $\E_n^*$ corresponding to a fixed frequency $\omega_n\in\W$ to be the solution to the Helmholtz equation
\begin{equation}
%\begin{cases}
\nabla\times\nabla\times\E_n^*(\bx,\omega_n) -(\K_0^n)^2 \E_n^*(\bx,\omega_n) = i\omega_n\mu_0\overline{\bdw}(\bx,\omega_n)\chi_{\Gamma}(\bx), \quad (\bx,\omega_n)\in \RR^3\times\RR,
%\end{cases}
\label{adjointN}
\end{equation}
and the \emph{single frequency phase conjugation} functional by
\begin{equation}
\Ic_n(\bx):=\ds\frac{\epsilon_0}{2\pi c_0\mu_0}\E_n^*(\bx,\omega_n).\label{I_n}
\end{equation}

Our first result of this section is the following Lemma.

\begin{lem}
For all $\bx\in\Omega$ sufficiently far from  $\Gamma$, compared to the wavelength of the wave impinging upon $\Omega$,  
$$
\Ic_n(\bx)\simeq\frac{\epsilon_0}{2\pi}\int_\Omega\Re e\bigg\{\hG^{\rm{ee}}_0(\bx-\by,\omega_n)\bigg\}\J(\by)d\by.
$$
\end{lem}

\begin{proof}
The proof is very similar to that of Theorem \ref{thm:full}.
Recall that 
\begin{eqnarray*}
\begin{cases}
\E^\ast_n(\bx,\omega_n) 
%&=&
%\ds\int_{\RR^d}\G^{\rm{ee}}_0(\xi-\bx,\omega_n)\overline{\bdw}(\xi,\omega_n)\chi_\Gamma(\xi)d\xi
%\\\nm
%&=&
=\ds\int_{\Gamma}\G^{\rm{ee}}_0(\xi-\bx,\omega_n)\overline{\bdw}(\xi,\omega_n)d\sigma(\xi),
\\\nm
\bdw(\xi,\omega_n)=\E(\xi,\omega_n)\bigg|_{\Gamma}=\ds\int_\Omega\G^{\rm{ee}}_0(\by-\xi,\omega_n)\J(\by)d\by\bigg|_{\xi\in\partial\Omega}.
\end{cases}
\end{eqnarray*}
Therefore, 
\begin{eqnarray*}
\Ic_n(\bx)&=&\ds\frac{\epsilon_0}{2\pi c_0\mu_0}\E_n^*(\bx,\omega_n),
%\\\nm
%&=&
%=
%\ds\frac{1}{2\pi c_0\mu_0^2}\int_\Gamma\G^{\rm{ee}}_0(\xi-\bx,\omega_n)\overline{\bdw}(\xi,\omega_n)d\sigma(\xi),
%\\\nm
%&=&
%=\ds\frac{1}{2\pi c_0\mu_0^2}\int_{\RR^d}\int_{\Gamma}\G^{\rm{ee}}_0(\xi-\bx,\omega)\overline{\G^{\rm{ee}}_0}(\by-\xi,\omega_n)\J(\by)\,d\sigma(\xi)d\by,
\\\nm
&=&
\ds\frac{\epsilon_0}{2\pi c_0\mu_0}
\int_\Omega\int_{\Gamma}\G^{\rm{ee}}_0(\xi-\bx,\omega_n)\overline{\G^{\rm{ee}}_0}(\by-\xi,\omega_n)d\sigma(\xi)\J(\by)d\by.
\end{eqnarray*}
Finally, we conclude, again using the electromagnetic Helmholtz-Kirchhoff identity from Lemma \ref{EM-HKI}.%, that 
%\begin{equation}
%\Ic_n(\bx)\simeq\frac{1}{2\pi\mu_0}\int_\Omega\Re e\bigg\{\G^{\rm{ee}}_0(\bx-\by,\omega_n)\bigg\}\J(\by)d\by.\label{eq:In}
%\end{equation}
\end{proof}

%Following is our principle result of this section.
%\begin{thm}[Initial Guess]
%Prove using quadrature formula for approximating integrals using a finite set of values that 
%$$\Ic_n(\bx)= \J(\bx)+ error (\omega)$$.
%and quantify $error(\omega)$. For that approximate 
%$$
%\ds\frac{1}{2\pi\mu_0}\int_\RR\Re e\bigg\{\G^{\rm{ee}}_0(\bx-\by,\omega_n)\bigg\}d\omega.
%$$
%\end{thm}

%\begin{proof}
%rara
%{proof}

\subsection{Regularization}\label{ss:reg}

In this section, an optimized reconstruction to the current source density is provided using $\bd_\W$. The aim is to explore an $l_1-$regularization to optimize the localization of the support of the source. 

The objective is to resolve the following optimization problem:
\begin{eqnarray}
&&\J_\lambda(\bx):=\ds\argmin_{\widehat{\J}\in\RR^3}\OM\big(\widehat{\J}\big)+\OR\big(\widehat{\J}\big),
\label{TV}
\\
&&\OM(\widehat{\J}) :=\ds\dfrac{1}{2N}\sum_{n=1}^N\left\|\Ic_n(\bx)-\frac{\epsilon_0}{2\pi}\int_\Omega\Re e\bigg\{\G^{\rm{ee}}_0(\bx-\by,\omega_n)\bigg\}\widehat{\J}(\by)d\by\right\|^2
\\
&&\OR(\widehat{\J}):=\ds\lambda\left\|\widehat{\J}(\bx)\right\|_{l_1},
\label{TV:R}
\end{eqnarray}
where  the first term  $\OM$ is the data fidelity term and the second term $\OR$ accounts for the $l_1-$regularization. It is precised that $\lambda$ is a regularization parameter controlling the relative weights of the two terms and provides a trade-off between fidelity to the measurements and noise sensitivity. Here $\|\cdot\|$ denotes the Euclidean norm in $\RR^3$.

\subsubsection{Fast Iterative-Shrinkage Thrasholding Algorithm with Backtracking}

The direct computations of the solution $\J_\lambda$ to the minimization problem \eqref{TV} is not trivial, indeed, the $l_1-$term is not smooth, at least not differentiable. Thus, in order to obtain $\J_\lambda$ explicitly, approximation schemes are indispensable. In order to do so, we follow Beck \& Teboulle \cite{Beck_teboulle} and use their \emph{fast iterative shrinkage thresholding algorithm with backtracking}. This method belongs to the class of split gradient descent iterative schemes with a global convergence rate $O(k^{-2})$, where $k$ is the iteration counter.

For any $\gamma>0$, define the quadratic approximation of the Lagrangian 
$$
\OL\left(\widehat{\J},\lambda \right) = \OM(\widehat{\J})+\OR(\widehat{\J})
$$
by 
\begin{eqnarray}
\OP_\gamma(\bx,\by):= \OM(\by)+\langle \bx-\by,\nabla\OM(\by)\rangle + \frac{\gamma}{2}\|\bx-\by\|^2+\OR(\bx) 
\end{eqnarray}
We also define 
\begin{eqnarray}
\OT_\gamma(\by):= \argmin_{\bx\in \RR^3}\bigg\{\OP_\gamma(\bx,\by)\bigg\},
\end{eqnarray}
where $\bx,\by\in\RR^3$ and $\langle \cdot,\cdot\rangle$ is the standard Euclidean inner product in $\RR^3$. Then the Algorithm \ref{EM:Itr} converges to the global minimum;  see \cite{Beck_teboulle}:
\begin{algorithm}
\caption{ Fast Iterative-Shrinkage Thresholding with Backtracking.}
\label{EM:Itr}
\begin{algorithmic}[1]
\Require Set $\gamma_0>0$,\qquad $\eta>1$,\qquad $\bx_0={\bf{0}}$,\qquad $\by_1=\bx_0$,\qquad $s_1=1$.

%\Procedure{Iterative-Shrinkage Thresholding.}{}
	
    \For {$k\geq 1$} 
   		\State	Set $i_k=1$,\qquad $\beta= \eta\gamma_{k-1}$.
   		
   		\While {$\ds\OL\left(\OT_{\beta}(\by_k),\lambda\right)>\OP_{\beta}\left(\OT_{\beta}(\by_k),\by_k\right)$} 
   		
   		\State 	Update $i_k=i_k+1$,\qquad $\beta= \eta^{i_k}\gamma_{k-1}$.
   		
   		\EndWhile 
   		
		\State Set $\gamma_k = \beta$,\qquad $x_k=\OT_{\gamma_k}(\by_k)$.
   			
  		\State Update $s_{k+1}= \dfrac{1}{2}\left(1+\sqrt{1+4s^2_k}\right)$, \quad $\by_{k+1} = \bx_k + \left(\dfrac{s_{k-1}}{s_{k+1}}\right) \left(\bx_k-\bx_{k-1}\right)$,\quad $i_k=0$, \quad  $k=k+1.$
   	 
   \EndFor

\end{algorithmic}

\Return $\widehat{\J}=\bx_{k}$.
\end{algorithm}

%\section{Numerical illustrations}\label{s:num}

%This section is dedicated to numerically ascertain the appositeness of the proposed source reconstruction paradigm.

%\begin{figure}[!h] 
%\label{fig:config}
%\centering
%\includegraphics[width=20pc]{config.eps} 
%\caption{Configuration} 
%\end{figure}

\section{Conclusion}\label{s:conc}
In this investigation, electromagnetic inverse source problem is tackled using boundary measurements of the tangential component of electric field over a finite set of frequencies. A phase conjugation algorithm is proposed in order to deal with the problem associated with full frequency spectrum which subsequently inspired an imaging functional for that with finite set of frequencies. Since the information is lost due to incomplete frequency spectrum, an $l_1-$regularization blended with the fast iterative shrinkage thresholding algorithm with backtracking of Beck and Teboulle is deployed. 

\bibliographystyle{plain}

\end{document}